\documentclass{article}
\usepackage[latin1]{inputenc}
\usepackage{times}
\usepackage{mathptmx}
\usepackage{bm}
\usepackage{latexsym}
\usepackage{amsmath}
\usepackage{amsthm}
\usepackage{amssymb}
\usepackage{graphicx}
\usepackage{ifthen}
\usepackage{calc}
\usepackage{a4wide}
\usepackage{stmaryrd}
\usepackage{cancel}
\usepackage{comment}

\makeatletter
\renewcommand{\section}{\@startsection
  {section}%
  {1}%
  {0mm}%
  {-1\baselineskip}%
  {0.5\baselineskip}%
  {\normalfont\large\bfseries}%
}
\renewcommand{\subsection}{\@startsection
  {subsection}%
  {2}%
  {0mm}%
  {-1\baselineskip}%
  {0.5\baselineskip}%
  {\normalfont\large\itshape}%
}

\renewcommand{\subsubsection}{\@startsection
  {subsubsection}%
  {3}%
  {0mm}%
  {-1\baselineskip}%
  {0.5\baselineskip}%
  {\normalfont\itshape}%
}

\newsavebox{\tempbox}
\renewcommand{\@makecaption}[2]{
  \vspace{10pt}
  \sbox{\tempbox}{\textbf{#1.} #2}
  \ifthenelse{\lengthtest{\wd\tempbox > \linewidth}}{
    \textbf{#1.} #2\par
  }{
    \begin{center}
      \textbf{#1.} #2
    \end{center}
  }
}

\makeatother

\numberwithin{equation}{section}

\numberwithin{figure}{section}

\usepackage{pstricks}
\usepackage{pst-node}
\usepackage{pst-tree}
\usepackage{pst-text}
\usepackage{rotating}

\psset{dotsize=5pt,linewidth=0.8pt}
\newpsobject{showgrid}{psgrid}{subgriddiv=1,griddots=10,gridlabels=6pt}

\newtheoremstyle{mythm}
  {}
  {}
  {\itshape}
  {}
  {\bfseries}
  {.}
  {.5em}
  {\thmname{#1}~\thmnumber{#2}\ifthenelse{\equal{\thmnote{#3}}{}}{}{~(\thmnote{#3})}}

\newtheoremstyle{mydefn}
  {}
  {}
  {\upshape}
  {}
  {\bfseries}
  {.}
  {.5em}
  {\thmname{#1}~\thmnumber{#2}\ifthenelse{\equal{\thmnote{#3}}{}}{}{~(\thmnote{#3})}}

\newtheoremstyle{myremark}
  {}
  {}
  {\upshape}
  {}
  {\itshape}
  {.}
  {.5em}
  {\thmname{#1}~\thmnumber{#2}\ifthenelse{\equal{\thmnote{#3}}{}}{}{~(\thmnote{#3})}}

\theoremstyle{mythm}
\newtheorem{theo}{Theorem}[section]
\newtheorem{lem}[theo]{Lemma}
\newtheorem{prop}[theo]{Proposition}
\newtheorem{cor}[theo]{Corollary}
\newtheorem{fact}[theo]{Fact}

\theoremstyle{mydefn}
\newtheorem{defn}[theo]{Definition}
\newtheorem{exa}[theo]{Example}
\theoremstyle{myremark}
\newtheorem{rem}[theo]{Remark}
\theoremstyle{mythm}
                      
\newenvironment{claim}[1]{
  \medskip\par
    \noindent\textit{Claim #1.}
}{
\par\medskip
}

\newcommand{\case}[1]{\par\medskip\noindent\textit{Case #1: }}

\newenvironment{cs}{
  \begin{list}{}{%
       \setlength{\labelwidth}{2em}%
       \setlength{\leftmargin}{\labelwidth}%
       \setlength{\itemindent}{0em}%
       \setlength{\labelsep}{0pt}
       \setlength{\listparindent}{1.5em}%
       \setlength{\parsep}{0em}%
    }
    \renewcommand{\case}[1]{\item[\textit{Case ##1: }]}
  }{
  \end{list}
}

\newcommand{\uend}{\hfill$\lrcorner$}

\newcounter{erom}
\renewcommand{\theerom}{(\roman{erom})}
\newenvironment{eroman}{
  \begin{list}{\theerom}{%
      \usecounter{erom}
      \setlength{\labelwidth}{3em}%
      \setlength{\leftmargin}{\labelwidth}%
      \setlength{\itemindent}{0em}%
      \setlength{\listparindent}{0em}%
      \setlength{\topsep}{5pt}%
      \setlength{\itemsep}{5pt}%
      \setlength{\parsep}{0pt}%
    }
  }{
  \end{list}
}

\newcounter{eal}
\renewcommand{\theeal}{(\Alph{eal})}

\newcounter{nlistcounter}
\newenvironment{nlist}[1]{
  \renewcommand{\thenlistcounter}{\upshape(#1.\arabic{nlistcounter})}
  \begin{list}{\bfseries\thenlistcounter}{%
      \usecounter{nlistcounter}
      \setlength{\labelwidth}{1.5em}%
      \setlength{\leftmargin}{\labelwidth}%
      \addtolength{\leftmargin}{\labelsep}%
      \setlength{\listparindent}{0em}%
      \setlength{\topsep}{5pt}%
      \setlength{\itemsep}{5pt}%
      \setlength{\parsep}{0pt}%
    }
  }{
  \end{list}
}

\renewcommand{\fboxsep}{2mm}

\newsavebox{\fminibox}
\newlength{\fminilength}
\newenvironment{fminipage}[1][\linewidth]{
  \setlength{\fminilength}{#1-2\fboxsep-2\fboxrule}%
  \begin{lrbox}{\fminibox}\begin{minipage}{\fminilength}}{
    \end{minipage}\end{lrbox}\noindent\fbox{\usebox{\fminibox}}
}

\newlength{\pwidth}

\newenvironment{problem}[3]{
  \begin{center}
    \ifthenelse{\equal{#1}{}}%
    {\setlength{\pwidth}{\textwidth-2\parindent}}%
    {\setlength{\pwidth}{#1cm}}

    \begin{fminipage}[\pwidth]\upshape
      \ifthenelse{\equal{#3}{}}{}{{\scshape #3}}
      \begin{list}{}{
          \ifthenelse{\equal{#2}{}}%
          {\settowidth{\labelwidth}{\textit{Parameter:}}}%
          {\settowidth{\labelwidth}{\textit{#2:}}}
          
          \setlength{\leftmargin}{\labelwidth+\labelsep}
          \setlength{\itemsep}{0ex}
          \setlength{\parsep}{0ex}
          \setlength{\topsep}{0.2ex}
        }
      }{
      \end{list}
    \end{fminipage}
  \end{center}
}

\renewcommand{\mathbf}[1]{\bm{#1}}

\renewcommand{\deg}{\operatorname{deg}}

\renewcommand{\phi}{\varphi}
\newcommand{\bigmid}{\;\big|\;}
\newcommand{\Bigmid}{\;\Big|\;}

\newcommand{\formel}[1]{\textsf{\upshape\small #1}}
\newcommand{\logic}[1]{\textsf{\upshape\small #1}}

\newcommand{\LL}{\logic L}

\newcommand{\ifp}{\operatorname{ifp}}

\newcommand{\LFP}{\logic{LFP}}
\newcommand{\IFP}{\logic{IFP}}
\newcommand{\IFPC}{\logic{\IFP+C}}
\newcommand{\FOL}{\logic{FO}}
\newcommand{\FOC}{\logic{FO+C}}

\newcommand{\PTIME}{\logic{PTIME}}
\newcommand{\NP}{\logic{NP}}
\newcommand{\NL}{\logic{NLOGSPACE}}
\newcommand{\LOGSPACE}{\logic{LOGSPACE}}
\newcommand{\IFPR}{\logic{IFP+R}}
\newcommand{\CPC}{\logic{CP+C}}

\newcommand{\NN}{\mathbb N_0}
\newcommand{\PN}{\mathbb N}

\newcommand{\dagle}{\trianglelefteq}
\newcommand{\dagsle}{\triangleleft}

\newcommand{\les}{\leqslant}

\newcommand{\slex}{\le_{\textup{s-lex}}}

\renewcommand{\int}{\operatorname{int}}

\newcommand{\cone}{\gamma}
\newcommand{\comp}{\alpha}
\newcommand{\bag}{\beta}
\newcommand{\sep}{\sigma}

\newcommand{\gapp}{\gamma_{\textit{app}}}

\newcommand{\CC}{\mathcal C}

\newcommand{\CG}{\mathcal G}

\newcommand{\CK}{\mathcal K}
\newcommand{\CL}{\mathcal L}

\newcommand{\CP}{\mathcal P}

\newcommand{\MCL}{\textit{MCL}}

\newcommand{\CRD}{\mbox{$\mathcal{C\!D}$}}
\newcommand{\CFI}{\mathcal P_{\textup{CFI}}}

\begin{document}
\title{Fixed-Point Definability and Polynomial Time on Chordal Graphs and Line Graphs}
\author{Martin Grohe\\\normalsize Humboldt University Berlin}
\date{}
\maketitle

\begin{abstract}
  The question of whether there is a logic that captures polynomial time was
  formulated by Yuri Gurevich in 1988. It is still wide open and regarded as one
  of the main open problems in finite model theory and database
  theory. Partial results have been obtained for specific classes of
  structures. In particular, it is known that fixed-point logic with counting
  captures polynomial time on all classes of graphs with excluded minors. The
  introductory part of this paper is a short survey of the state-of-the-art in 
  the quest for a logic capturing polynomial time. 

  The main part of the paper is concerned with classes of graphs defined by
  excluding induced subgraphs. Two of the most fundamental such classes are
  the class of chordal graphs and the class of line graphs. We prove that
  capturing polynomial time on either of these classes is as hard as capturing
  it on the class of all graphs. In particular, this implies that fixed-point
  logic with counting does not capture polynomial time on these classes. Then
  we prove that fixed-point logic with counting does capture polynomial time
  on the class of all graphs that are both chordal and line graphs.
\end{abstract}

\section{The quest for a logic capturing PTIME}
\label{sec:survey}
Descriptive complexity theory started with Fagin's Theorem~\cite{fag74}
from 1974, stating that existential second-order logic \emph{captures}
the complexity class \NP. This means that a property of finite
structures is decidable in nondeterministic polynomial time if and
only if it is definable in existential second order logic. Similar
logical characterisations where later found for most other complexity
classes. For example, in 1982 Immerman~\cite{imm86} and independently
Vardi~\cite{var82} characterised the class \PTIME\ (polynomial time)
in terms of least fixed-point logic, and in 1983 Immerman~\cite{imm87}
characterised the classes \NL\ (nondeterministic logarithmic
space) and \LOGSPACE\ (logarithmic space) in terms of transitive
closure logic and its deterministic variant. However, these logical
characterisations of the classes \PTIME, \NL, and \LOGSPACE, and all
other known logical characterisations of complexity classes contained
in \PTIME, have a serious drawback: They only apply to properties of
\emph{ordered structures}, that is, relational structures with one
distinguished relation that is a linear order of the elements of the
structure. It is still an open question whether there are logics that
characterise these complexity classes on arbitrary, not necessarily
ordered structures. We focus on the class \PTIME\ from now on. In this
section, which is an updated version of \cite{gro08b}, we give a short survey of
the quest for a logic capturing PTIME.

\subsection{Logics capturing PTIME}\label{sec:lcp}
The question of whether there is a logic that characterises, or
\emph{captures}, \PTIME\ is subtle. If phrased naively, it has a
trivial, but completely uninteresting positive
answer. Yuri Gurevich~\cite{gur88} was the first to give a precise
formulation of the question. Instead of arbitrary finite structures,
we restrict our attention to graphs in this paper. This is no serious
restriction, because the question of whether there is a logic that
captures \PTIME\ on arbitrary structures is 
equivalent to the restriction of the question to graphs. We first need
to define what constitutes a logic. Following Gurevich, we take a very liberal,
semantically oriented approach. We identify \emph{properties} of
graphs with classes of graphs closed under isomorphism. A logic \LL\ (on graphs)
consists of a computable set of \emph{sentences} together with a
semantics that associates a property $\CP_\phi$ of graphs with each
sentence $\phi$. We say that a graph $G$ \emph{satisfies} a sentence
$\phi$, and write $G\models\phi$, if $G\in\CP_\phi$. We say that a
property $\CP$ of graphs is \emph{definable} in $\LL$ if there is a
sentence $\phi$ such that $\CP_\phi=\CP$.
A logic $\LL$
\emph{captures} \PTIME\ if the following two conditions are satisfied:
\begin{nlist}{G}
\item\label{li:c1} Every property of graphs that is decidable in
  \PTIME\ is definable in \LL.
\item\label{li:c2} There is a computable function that associates with
  every \LL-sentence $\phi$ a polynomial $p(X)$ and an algorithm $A$
  such that $A$ decides the property $\CP_\phi$ in time $p(n)$, where
  $n$ is the number of vertices of the input graph.
\end{nlist}
While condition \ref{li:c1} is obviously necessary, condition \ref{li:c2} may
seem unnecessarily complicated. The natural condition we expect to see instead
is the following condition (G.2'): Every property of graphs that is definable
in $\LL$ is decidable in \PTIME.  Note that (G.2) implies (G.2'), but that the
converse does not hold. However, (G.2') is too weak, as the following example
illustrates:

\begin{exa}
  Let 
  $\CP_1,\CP_2,\ldots$ be an arbitrary enumeration of
  all polynomial time decidable properties of graphs. Such an
  enumeration exists because there are only countably many Turing
  machines and hence only countably many decidable properties of graphs.
  Let $\LL'$ be the ``logic'' whose sentences are the natural numbers
  and whose semantics is defined by letting sentence $i$ define
  property $\CP_i$. Then $\LL'$ is a logic according to our
  definition, and it does satisfy (G.1) and (G.2'). But clearly,
  $\LL'$ is not a ``logic capturing \PTIME'' in any interesting sense.
\end{exa}

Let me remark that most natural logics that are candidates
for capturing \PTIME\ trivially satisfy (G.2). The difficulty is to
prove that they also satisfy (G.1), that is, define all \PTIME-properties.

There is a different route that leads to the same question of whether
there is a logic capturing \PTIME\ from a database-theory perspective:
After Aho and Ullman~\cite{ahoull79} had realised that SQL, the
standard query language for relational databases, cannot express all database
queries computable in polynomial time, Chandra and
Harel~\cite{chahar82} asked for a recursive enumeration of the class of all
relational database queries computable in polynomial time.  It turned
out that Chandra and Harel's question is equivalent to Gurevich's
question for a logic capturing \PTIME, up to a minor technical
detail.\footnote{In Chandra and Harel's version of the question,
  condition (G.2) needs to be replaced by the following condition (CH.2):
There is a computable function that associates with every \LL-sentence
  $\phi$ an algorithm $A$ such that $A$
  decides the property $\CP_\phi$ in polynomial time. The difference
  between (G.2) and (CH.2) is that in (CH.2) the polynomial bounding
  the running time of the algorithm $A$ is not required to be
  computable from $\phi$.}

The question of whether there is a logic that captures \PTIME\ is
still wide open, and it is considered one of the main open problems in
finite model theory and database theory. Gurevich conjectured
that there is no logic capturing \PTIME. This would not only imply
that $\PTIME\neq\NP$ --- remember that by Fagin's Theorem there is a
logic capturing \NP\ --- but it would actually have interesting
consequences for the structure of the complexity class
\PTIME. Dawar~\cite{daw95} proved a dichotomy theorem stating that, depending on the answer to the
question, there are two fundamentally different possibilities: If
there is a logic for \PTIME, then the structure of $\PTIME$ is very simple; all \PTIME-properties
are variants or special cases of just one problem. If there is no
logic for \PTIME, then the structure of \PTIME\ is so complicated that it eludes all attempts for
a classification. The formal statement of the first possibility is that
there is a complete problem for \PTIME\ under first-order
reductions. The formal statement of the second possibility is that the
class of \PTIME-properties is not recursively enumerable.\footnote{The version
  of recursive enumerability used here is not exactly the same as the one
  considered by Chandra and Harel~\cite{chahar82}; the difference is
  essentially the same as the difference between conditions (G.2) and (CH.2)
  discussed earlier.}

\subsection{Fixed-point logics}
Fixed-point logics play an important role in finite-model theory, and in
particular in the quest for a logic capturing \PTIME. Very briefly, the
fixed-point logics considered in this context are extensions of first-order
logic by operators that formalise inductive definitions. We have already
mentioned that \emph{least fixed-point logic} \LFP\ captures polynomial time
on ordered structures; this result is known as the \emph{Immerman-Vardi
Theorem}. For us, it will be more convenient to work with \emph{inflationary
  fixed-point logic} \IFP, which was shown to have the same expressive power
as \LFP\ on finite structures by Gurevich and Shelah~\cite{gurshe86} and on
infinite structures by Kreutzer~\cite{kre04}.

\IFP\ does not capture polynomial time on all finite structures. The most
immediate reason is the inability of the logic to count. For example, there is
no \IFP-sentence stating that the vertex set of a graph has even cardinality;
obviously, the graph property of having an even number of vertices is
decidable in polynomial time. This led Immerman~\cite{imm87a} to extending
fixed-point logic by ``counting operators''. The formal definition of
fixed-point logic with counting operators that we use today,
\emph{inflationary fixed-point logic with counting} \IFPC, is due to Grädel
and Otto~\cite{graott99}. \IFPC\ comes surprisingly close to
capturing \PTIME. Even though Cai, Fürer, and
Immerman~\cite{caifurimm92} gave an example of a property of graphs that is
decidable in \PTIME, but not definable in \IFPC, it turns out that the logic
does capture \PTIME\ on many interesting classes of structures.

\subsection{Capturing PTIME on classes of graphs}

Let $\CC$ be a class of graphs,
which we assume to be closed under isomorphism. We say that a logic \LL\ captures \PTIME\ \emph{on $\CC$} if
it satisfies the following two conditions:
\begin{nlist}{G}
\item[\bfseries (G.1)$_{\CC}$] For every property $\CP$ of graphs that is decidable in \PTIME\ there is an
  \LL-sentence $\phi$ such that for all graphs $G\in\mathcal C$ it holds
  that $G\models\phi$ if and only if $G\in\CP$.
\item[\bfseries (G.2)$_{\CC}$] There is a computable function that associates with every \LL-sentence
  $\phi$ a polynomial $p(X)$ and an algorithm $A$ such that given
  a graph $G\in\CC$, the algorithm $A$
  decides if $G\models\phi$ in time $p(n)$, where $n$ is the number of
  vertices of $G$.
\end{nlist}
Note that these conditions coincide with conditions (G.1) and (G.2) if $\CC$ is
the class of all graphs.

The first positive result in this direction is due to Immerman and
Lander~\cite{immlan90}, who proved that \IFPC\ captures \PTIME\ on the class
of all trees. In 1998, I proved that $\IFPC$ captures $\PTIME$ on the class of
all planar graphs~\cite{gro98a} and around the same time, Julian Mari{\~{n}}o
and I proved that $\IFPC$ captures $\PTIME$ on all classes of structures of
bounded tree width \cite{gromar99}. In~\cite{gro08a}, I proved the same result
for the class of all graphs that
have no complete graph on five vertices, $K_5$, as a minor. A \emph{minor} of graph
$G$ is a graph $H$ that can be obtained from a subgraph of $G$ by contracting
edges. We say that a class $\CC$ of
graphs \emph{excludes a minor} if there is a graph $H$ that is not a minor of
any graph in $\CC$. Very recently, I proved that $\IFPC$ captures $\PTIME$ on
all classes of graphs that exclude a minor~\cite{gro10}.

In the last few years, maybe as a consequence of Chudnowsky,
Robertson, Seymour, and Thomas's \cite{crst06} proof of the strong
perfect graph theorem, the focus of many graph theorists has shifted
from graph classes with excluded minors to graph classes defined by
excluding induced subgraphs. One of the most basic and important example of
such a class is the class of \emph{chordal graphs}. A cycle $C$ of a
graph $G$ is \emph{chordless} if it is an induced subgraph. A graph is
\emph{chordal} (or \emph{triangulated}) if it has no chordless cycle of length
at least four. Figure~\ref{fig:chordal-line}(a) shows an example of a
chordal graph. All chordal graphs are
\emph{perfect}, which means that the graphs themselves and all their
induced subgraphs have the chromatic number equal to the clique number. Chordal
graphs have a nice and simple structure; they can be decomposed into a
tree of cliques. A second important example is the class of \emph{line
  graphs}. The line graph of a graph $G$ is the graph $L(G)$ whose
vertices are the edges of $G$, with two edges being adjacent in $L(G)$
if they have a common endvertex in $G$. Figure~\ref{fig:chordal-line}(b) shows an example of a
line graph. The class of all line graphs
is closed under taking induced subgraphs. Beineke~\cite{bei70} gave a
characterisation of the class of line graphs (more
precisely, the class of all graphs isomorphic to a line graph) by a
family of nine excluded subgraphs. An extension of the class of line
graphs, which has also received a lot of attention in the literature, is the class of \emph{claw-free} graphs. A graph is claw-free if it does
not have a vertex with three pairwise nonadjacent neighbours, that is, if it
does not have a \emph{claw} (displayed in Figure~\ref{fig:claw}) as an induced
subgraph. It is easy to see that all line graphs are claw-free. Recently,
Chudnowsky and Seymour (see \cite{chusey05}) developed a structure theory for
claw-free graphs.
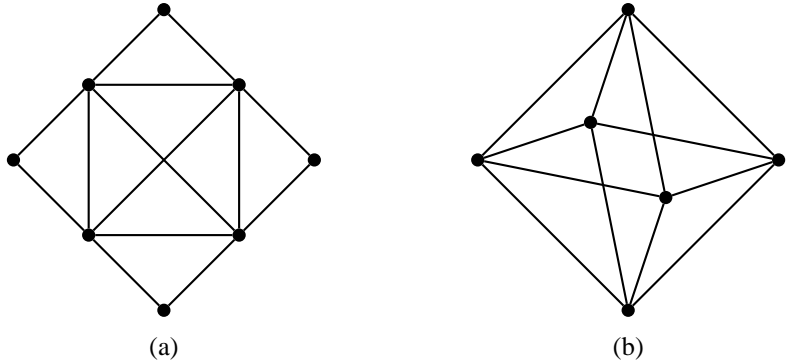
\begin{figure}
  \centering
  \begin{pspicture}(0,-0.5)(4,4)
    \dotnode(1,3){A}
    \dotnode(3,3){B}
    \dotnode(3,1){C}
    \dotnode(1,1){D}
    \dotnode(2,4){E}
    \dotnode(4,2){F}
    \dotnode(2,0){G}
    \dotnode(0,2){H}
    \ncline{A}{B}
    \ncline{A}{C}
    \ncline{A}{D}
    \ncline{B}{C}
    \ncline{B}{D}
    \ncline{C}{D}
    \ncline{A}{E}
    \ncline{B}{E}
    \ncline{B}{F}
    \ncline{C}{F}
    \ncline{C}{G}
    \ncline{D}{G}
    \ncline{D}{H}
    \ncline{A}{H}

    \rput(2,-0.5){(a)}
  \end{pspicture}
  \hspace{2cm}
  \begin{pspicture}(0,-0.5)(4,4)
    \dotnode(2,4){AB}
    \dotnode(4,2){BC}
    \dotnode(2,0){CD}
    \dotnode(0,2){AD}
    \dotnode(1.5,2.5){AC}
    \dotnode(2.5,1.5){BD}

    \ncline{AD}{AB}
    \ncline{AB}{BC}
    \ncline{BC}{CD}
    \ncline{CD}{AD}
    \ncline{AC}{AB}
    \ncline{AC}{AD}
    \ncline{AC}{CD}
    \ncline{AC}{BC}
    \ncline{BD}{AB}
    \ncline{BD}{AD}
    \ncline{BD}{CD}
    \ncline{BD}{BC}

    \rput(2,-0.5){(b)}
    
  \end{pspicture}

  \caption{(a) a chordal graph, which is not a line graph, and (b) the line
    graph of $K_4$, which is not chordal}
  \label{fig:chordal-line}
\end{figure}

\begin{figure}
  \centering
  \begin{pspicture}(2,1)
    \dotnode(0.2,0.2){A}
    \dotnode(1,0){B}
    \dotnode(1.8,0.2){C}
    \dotnode(1,1){D}
    \ncline{A}{D}
    \ncline{B}{D}
    \ncline{C}{D}
  \end{pspicture}
  \hspace{1cm}

  \caption{A claw}
  \label{fig:claw}
\end{figure}

It would be tempting to use this structure theory for claw free graphs, or at
least the simple treelike structure of chordal graphs, to prove that \IFPC\ captures \PTIME\ on these classes in a similar way as the structure theory for
classes of graphs with excluded minors is used to prove that \IFPC\ captures
\PTIME\ on classes with excluded minors. Unfortunately, this is only possible
on the very restricted class of graphs that are both chordal and line
graphs (an example of such a graph is shown in Figure~\ref{fig:cl} on p.\pageref{fig:cl}). We prove the following theorem:

\begin{samepage}
\begin{theo}~
  \begin{enumerate}
  \item
    \IFPC\ does not capture \PTIME\ on the class of chordal graphs or
    on the class of line graphs.
  \item \IFPC\ captures \PTIME\ on the class of chordal line graphs.
    \end{enumerate}
\end{theo}
\end{samepage}

Our construction to prove (1) is so simple that it will
apply to any reasonable logic, which means that if a ``reasonable''
logic captures \PTIME\ on the class of chordal graphs or on the class
of line graphs, then it captures \PTIME\ on the class of all
graphs. 

Further interesting graph classes closed under taking induced subgraphs
are various classes of intersection graphs. Very recently,
Laubner~\cite{lau10} proved that \IFPC\ captures \PTIME\ on the
class of all interval graphs. To conclude our discussion of classes of graphs on which \IFPC\ captures
\PTIME, let me mention a result due to Hella, Kolaitis, and
Luosto~\cite{helkolluo96} stating that \IFPC\ captures \PTIME\ on almost all
graphs (in a precise technical sense). Thus it seems that the results for
specific classes of graphs are not very surprising, but it should be
mentioned that almost no graphs fall in one of the natural graphs
classes discussed before.

Instead of capturing all \PTIME\ on a specific class of structures,
Otto~\cite{ott97a,ott97c,ott99} studied the question of capturing all \PTIME\
properties satisfying certain invariance conditions. Most notably, he proved
that bisimulation-invariant properties are decidable in polynomial time
if and only if they are definable in the \emph{higher-dimensional
  $\mu$-calculus}.

\subsection{Isomorphism testing and canonisation}

As an abstract question, the question of whether there is a logic capturing polynomial time is
linked to the graph isomorphism and canonisation problems. Otto~\cite{ott97a} was the first to systematically study the connection
between canonisation and descriptive complexity theory. Specifically, if
there is a polynomial time canonisation algorithm for a class
$\CC$ of graphs, then there is a logic that captures
polynomial time on this class $\CC$. This follows from the
Immerman-Vardi Theorem. To explain it, let us assume
that we represent graphs by their adjacency matrices. A \emph{canonisation
  mapping} gets as argument some adjacency matrix representing a graph and
returns a \emph{canonical} adjacency matrix for this graph, that is, it maps
\emph{isomorphic} adjacency matrices to \emph{equal} adjacency matrices. 
 As an
adjacency matrix for a graph is completely fixed once we specify the ordering
of the rows and columns of the matrix, we may view a canonisation as a mapping
associating with each graph a canonical ordered copy of the graph. Now we can
apply the Immerman-Vardi Theorem to this ordered copy.

Clearly, if there is a polynomial time canonisation mapping for a class of
graphs (or other structures) then there is a polynomial time isomorphism test
for this class. It is open whether the converse also holds. It is also open
whether the existence of a logic for polynomial time implies the existence of
a polynomial time isomorphism test or canonisation mapping.

Polynomial time canonisation mappings are known for many natural classes of
graphs, for example planar graphs \cite{hoptar72,hopwon74}, graphs of bounded
genus~\cite{filmay80,mil80}, graphs of bounded eigenvalue
multiplicity~\cite{babgrimou82},
graphs of bounded
degree~\cite{babluk83,luk82}, and graphs of bounded tree width~\cite{bod90}.
Hence for all theses classes there are logics capturing \PTIME.
However, the logics obtained through canonisation hardly qualify as natural
logics. If a logic is to contribute to our understanding of the complexity
class \PTIME --- and from my perspective this is the main reason for being
interested in such a logic --- we have to look for natural logics that derive
their expressiveness from clearly visible basic principles like inductive
definability, counting or other combinatorial operations, and maybe fundamental
algebraic operations like computing the rank or the determinant of a
matrix. If such a logic captures polynomial time on a class of structures,
then this shows that all polynomial time properties of structures in this
class are based on the principles underlying the logic. Thus even for classes
for which we know that there is a logic capturing \PTIME\ through a
polynomial-time canonisation algorithm, I think it is important to find
``natural'' logics capturing \PTIME\ on these classes. In particular,
I view it as an
important open problem to find a natural logic that captures \PTIME\ on
classes of graphs of bounded degree. It is known that \IFPC\ does not capture
\PTIME\ on the class of all graphs of maximum degree at most three.

Most known capturing results are proved by showing that there is a
canonisation mapping that is definable in some logic. In particular, all
capturing results for \IFPC\ mentioned above are proved this way. 
It was
observed by Cai, Fürer, and Immerman \cite{caifurimm92} that for classes
$\CC$ of structures which admit a canonisation mapping definable in
\IFPC, a simple combinatorial algorithm known as the Weisfeiler-Lehman (WL)
algorithm \cite{evdkarpon99,evdpon99} can be used as a polynomial time
isomorphism test on $\mathcal C$. Thus the the WL-algorithm correctly decides
isomorphism on the class of chordal line graphs and on all classes of
graphs with excluded minors. A refined version of the same approach was used by
Verbitsky and others \cite{grover06,kobver08,ver07} to obtain parallel
isomorphism tests running in polylogarithmic time for planar graphs and graphs
of bounded tree width.

\subsection{Stronger logics}
Early on, a number of results regarding the possibility of capturing
polynomial time by adding Lindström quantifiers to first-order logic or
fixed-point logic were obtained. Hella \cite{hel89} proved that adding
finitely many Lindström quantifiers (or infinitely many of bounded arity) to
fixed-point logic does not suffice to capture polynomial time (also see
\cite{dawhel94}). Dawar~\cite{daw95a} proved that if there is a logic
capturing polynomial time, then there is such a logic obtained from
fixed-point logic by adding one vectorised family of Lindström quantifiers.
Another family of logics that have been studied in this context consists of extensions
of fixed-point logic with
nondeterministic choice operators \cite{abisimvia90,dawric03,girhoa98}.

Currently, the two main candidates for logics capturing $\PTIME$
are \emph{choiceless polynomial time with counting $\CPC$} and
\emph{inflationary fixed-point logic with a rank operator $\IFPR$}.  The logic $\CPC$ was
introduced by Blass, Gurevich and Shelah~\cite{blagurshe99}
(also see~\cite{blagurshe02,dawricros06}). The formal definition of the logic
is carried out in the framework of \emph{abstract state machines} (see, for
example, \cite{gur00}). Intuitively
$\CPC$ may be viewed as a version of $\IFPC$ where quantification and fixed-point
operators not only range over elements of a structure, but instead over all
objects that can be described by $O(\log n)$ bits, where $n$ is the size of
the structure. This intuition can be formalised in an expansion of a structure
by all hereditarily finite sets which use the elements of the structure as
atoms. The logic $\IFPR$ \cite{dawgrohollau09} is an
extension of $\IFP$ by an operator that determines the rank of definable
matrices in a structure. This may be viewed as a higher
dimensional version of a counting operator. (Counting appears as a special
case of diagonal $\{0,1\}$-matrices.)

Both $\CPC$ and $\IFPR$ are known to be strictly more expressive than $\IFPC$. Indeed, both logics can express the property used by Cai, Fürer, and Immerman
to separate $\IFPC$ from $\PTIME$. For both logics it is open whether they
capture polynomial time, and it is also open whether one of them semantically contains
the other.

\section{Preliminaries}

$\NN$, and $\PN$ denote the sets
of nonnegative integers and natural numbers (that is, positive integers),
respectively. 
For $m,n\in\NN$, we let $[m,n]:=\{\ell\in\NN\mid m\le\ell\le n\}$ and
$[n]:=[1,n]$.  We denote the power set of a set $S$ by $2^S$ and the
set of all $k$-element subsets of $S$ by $\binom{S}{k}$. 

We often denote tuples $(v_1,\ldots,v_k)$ by $\vec v$. If $\vec v$
denotes the tuple $(v_1,\ldots,v_k)$, then by $\tilde v$ we denote the
set $\{v_1,\ldots,v_k\}$. If $\vec v=(v_1,\ldots,v_k)$ and $\vec
w=(w_1,\ldots,w_\ell)$, then by $\vec v\vec w$ we denote the tuple
$(v_1,\ldots,v_k,w_1,\ldots,w_\ell)$. By $|\vec v|$ we denote the length of a
tuple $\vec v$, that is, $|(v_1,\ldots,v_k)|=k$.

\subsection{Graphs}
\label{sec:gra}
Graphs in this paper are always finite, nonempty, and simple, where simple
means that there are no loops or parallel edges. Unless explicitly called
``directed'', graphs are undirected. The vertex set of a graph $G$ is denoted
by $V(G)$ and the edge set by $E(G)$. We view graphs as relational structures
with $E(G)$ being a binary relation on $V(G)$. However, we often find it
convenient to view edges (of undirected graphs) as 2-element subsets of $V(G)$
and use notations like $e=\{u,v\}$ and $v\in e$. Subgraphs, induced subgraphs, union, and intersection of
graphs are defined in the usual way. We write $G[W]$ to denote the induced
subgraph of $G$ with vertex set $W\subseteq V(G)$, and we write $G\setminus W$
to denote $G[V(G)\setminus W]$. The set $\{w\in V(G)\mid \{v,w\}\in
E(G)\}$ of \emph{neighbours} of a node $v$ is denoted by $N^G(v)$, or just $N(v)$ if $G$ is
clear from the context, and the \emph{degree} of $v$ is the
cardinality of $N(v)$. The \emph{order} of a
graph, denoted by $|G|$, is the number of vertices of $G$. The class
of all graphs is denoted by $\CG$. A \emph{homomorphism} from a graph
$G$ to a graph $H$ is a mapping $h:V(G)\to V(H)$ that preserves
adjacency, and an \emph{isomorphism} is a bijective homomorphism whose
inverse is also a homomorphism.

For every finite nonempty set $V$, we let
$K[V]$ be the \emph{complete graph} with vertex set $V$, and we let
$K_n:=K\big[[n]\big]$. A \emph{clique} in a graph $G$ is a set $W\subseteq V(G)$
such that $G[W]$ is a complete graph.
\emph{Paths} and \emph{cycles} in graphs are defined in the usual way. The
\emph{length} of a path or cycle is the number of its edges.
\emph{Connectedness}
and \emph{connected} components are defined in the usual way. A set
$W\subseteq V(G)$ is \emph{connected} in a graph $G$ if $W\neq\emptyset$ and
$G[W]$ is connected. For sets
$W_1,W_2\subseteq V(G)$, a set $S\subset V(G)$ \emph{separates $W_1$ from
  $W_2$} if there is no path from a
vertex in $W_1\setminus S$ to vertex in $W_2\setminus S$ in the graph
$G\setminus S$.  
 
A \emph{forest} is an undirected acyclic
graph, and a \emph{tree} is a connected forest.
It will be a useful convention to call the vertices of trees and forests
\emph{nodes}. 
A \emph{rooted tree} is a triple
$T=(V(T),E(T),r(T))$, where $(V(T),E(T))$ is a tree and $r(T)\in V(T)$ is a
distinguished node called the \emph{root}.

We occasionally have to deal with \emph{directed graphs}.  We allow
directed graphs to have loops. We use standard graph theoretic
terminology for directed graphs, without going through it in
detail. Homomorphisms and isomorphisms of directed graphs preserve the
direction of the edges. Paths and cycles in a directed graph are always meant to be
directed; otherwise we will call them ``paths or cycles of the
underlying undirected graph''. Note that cycles in directed graphs may
have length $1$ or $2$. For a directed graph $D$ and a vertex $v\in
V(D)$, we let $N^D(v):=\big\{w\in V(D)\bigmid (v,w)\in
E(D)\big\}$. \emph{Directed acyclic graphs} will be of particular
importance in this paper, and we introduce some additional terminology
for them: Let $D$ be a directed acyclic graph. A node $w$ is a
\emph{child} of a node $v$, and $v$ is a \emph{parent} of $w$, if
$(v,w)\in E(D)$. We let $\dagle^D$ be the reflexive transitive closure
of the edge relation $E(D)$ and $\dagsle^D$ its irreflexive
version. Then $\dagle^D$ is a partial order on $V(D)$. 

A \emph{directed tree} is a directed acyclic graph $T$ in which every node
has at most one parent, and for which there is a vertex $r$ called
the \emph{root} such that for all $t\in V(t)$ there is a path from $r$
to $t$. There is an obvious one-to-one correspondence between rooted
trees and directed trees: For a rooted tree $T$ with root $r:=r(T)$ we
define the corresponding directed tree $T'$ by $V(T'):=V(T)$ and
$E(T'):=\big\{(t,u)\bigmid \{t,u\}\in E(T)$ and $t$ occurs on the
  path $rTu\big\}$. We freely jump back and forth between rooted
trees and directed trees, depending on which will be
more convenient. In particular, we use the terminology introduced for
directed acyclic graphs (parents, children, the partial order $\dagle$, et cetera) for
rooted trees.

\subsection{Relational structures}

A \emph{relational structure} $A$ consists of a finite set $V(A)$ called the
\emph{universe} or \emph{vertex set} of $A$ and finitely many relations on
$A$. The only types of structures we will use in this paper are \emph{graphs},
viewed as structures $G=\big(V(G),E(G)\big)$ with one binary relation $E(G)$, and
\emph{ordered graphs}, viewed as structures $G=\big(V(G),E(G),\les(G)\big)$ with two binary relations $E(G)$
and $\les(G)$, where $\big(V(G),E(G)\big)$ is a graph and $\les(G)$ is a
linear order of the vertex set $V(G)$.

\subsection{Logics}

We assume that the reader has a basic knowledge in logic. In this section, we
will informally introduce the two main logics \IFP\ and \IFPC\ used in this
paper. For background and a precise definition, I refer the reader to one of the
textbooks~\cite{ebbflu99,gklmsvvw07,imm99,lib04}.  It will be convenient to
start by briefly reviewing \emph{first-order logic} $\FOL$. Formulae of
first-order logic in the language of graphs are built from atomic formulae
$E(x,y)$ and $x=y$, expressing adjacency and equality of vertices, by the usual
Boolean connectives and existential and universal quantifiers ranging over the
vertices of a graph. First-order formulae in the language of ordered graphs
may also contain atomic formulae of the form $x\les y$ with the obvious
meaning, and formulae in other languages may contain atomic formulae defined
for these languages. We
write $\phi(x_1,\ldots,x_k)$ to denote that the free variables of a formula
$\phi$ are among $x_1,\ldots,x_k$. For a graph $G$ and vertices
$v_1,\ldots,v_k$, we write $G\models\phi[v_1,\ldots,v_k]$ to denote that $G$
satisfies $\phi$ if $x_i$ is interpreted by $v_i$, for all $i\in[k]$. 

\emph{Inflationary fixed-point logic} \IFP\ is the extension of \FOL\ by a
fixed-point operator with an inflationary semantics. To introduce this operator,
let $\phi(X,\vec x)$ be a formula that, besides a $k$-tuple $\vec
x=(x_1,\ldots,x_k)$ of free \emph{individual
  variables} ranging over the vertices of a graph, has a free
$k$-ary \emph{relation variable} ranging over $k$-ary relations on the vertex
set. For every graph $G$ we
define a sequence $R_i=R_i(G,\phi,X,\vec x)$, for
$i\in\NN$, of $k$-ary relations on $V(G)$ as follows:
\begin{align*}
  R_0&:=\emptyset\\
  R_{i+1}&:=R_i\cup\big\{\vec v\bigmid G\models\phi[R_i,\vec v]\big\}&\text{for all }i\in\NN.
\end{align*}
Since we have $R_0\subseteq R_1\subseteq R_2\subseteq\cdots\subseteq V(G)^k$
and $V(G)$ is finite, the sequence reaches a fixed-point $R_n=R_{n+1}=R_i$ for
all $i\ge n$, which we denote by $R_\infty=R_\infty(G,\phi,X,\vec x)$. The
\emph{ifp-operator} applied to $\phi,X,\vec x$ defines this fixed-point. We
use the following syntax:
\begin{equation}
  \label{eq:ifp-formula}
\underbrace{\ifp\big(X\gets\vec x\bigmid \phi\big)\vec x'}_{=:\psi(\vec x')}.
\end{equation}
Here $\vec x'$ is another $k$-tuple of individual variables, which may
coincide with $\vec x$. The variables in the tuple $\vec x'$ are the free
variables of the formula $\psi(\vec x')$, and for every graph $G$ and every
tuple $\vec v\in V(G)^k$ of vertices we let $G\models\psi[\vec v]\iff\vec v\in
R_\infty$. These definitions can easily be
extended to a situation where the formula $\phi$ contains other free variables
than $X$ and and the variables in $\tilde x$; these variables remain free variables of
$\psi$. Now formulae of
inflationary fixed-point logic \IFP\ in the language of graphs are built from atomic formulae
$E(x,y)$, $x=y$, and $X\vec x$ for relation variables $X$ and tuples of
individual variables $\vec x$ whose length matches the arity of $X$, by the usual
Boolean connectives and existential and universal quantifiers ranging over the
vertices of a graph, and the ifp-operator.

\begin{exa}\label{exa:conn}
  The \IFP-sentence 
  \[
  \formel{conn}:=\forall x_1\forall x_2\;\ifp\Big(X\gets (x_1,x_2)\Bigmid
       x_1=x_2
         \vee E(x_1,x_2)
         \vee \exists x_3\big(X(x_1,x_3)\wedge X(x_3,x_2)\big)\Big)
       (x_1,x_2)
       \]
       states that a graph is connected.  
\end{exa}

\emph{Inflationary fixed-point logic with counting}, \IFPC, is the extension  of
\IFP\ by counting operators that allow it to speak about cardinalities of
definable sets and relations. To define \IFPC, we interpret the logic \IFP\
over two sorted extensions of graphs (or other relational structures) by a
numerical sort. For a graph $G$, we let $N(G)$ be the initial segment
$\big[0,|G|\big]$ of the nonnegative integers. We let $G^+$ be the two-sorted
structure $G\cup
(N(G),\le)$, where $\le$ is the natural linear order on $N(G)$. To avoid
confusion, we always assume that $V(G)$ and $N(G)$ are disjoint.
We call the elements of the first sort $V(G)$ \emph{vertices} and
the elements of the second sort $N(G)$ \emph{numbers}. Individual variables
of our logic range either over the set $V(G)$ of vertices of $G$ or over the set $N(G)$ of numbers of $G$. Relation variables may range over mixed relations, having
certain places for vertices and certain places for numbers. 
Let us call the resulting logic, inflationary fixed-point logic
over the two-sorted extensions of graphs, $\IFP^+$.  We may still view
$\IFP^+$ as a logic over plain graphs, because the
extension $G^+$ is uniquely determined by
$G$. More precisely, we say that a sentence $\phi$ of $\IFP^+$ is satisfied by a graph $G$ if
it $G^+\models\phi$.
\emph{Inflationary fixed-point logic with counting} \IFPC\ is the
extension of $\IFP^+$ by \emph{counting terms} formed as follows: For every
formula $\phi$ and every vertex variable $x$ we add a term $\#
x\;\phi$; the value of this term is the number of assignments to $x$
such that $\phi$ is satisfied.

With each \IFPC-sentence $\phi$ in the language of graphs we associate the
graph property $\CP_{\phi}:=\{G\mid G\models\phi\}$. As the set of all
\IFPC-sentences is computable, we may thus view \IFPC\ as an abstract logic
according to the definition given in Section~\ref{sec:lcp}. It is easy to see
that \IFPC\ satisfies condition \ref{li:c2} and therefore condition (G.2)$_{\CC}$
for every class $\CC$ of graphs. Thus to prove that \IFPC\ captures \PTIME\ on
a class $\CC$ it suffices to verify (G.1)$_{\CC}$.

In the following examples, we use the notational convention that $x$ and
variants such as $x_1,x'$ denote vertex variables and that $y$ and variants
denote number variables.

\begin{exa}\label{exa:even}
  The $\IFPC$-term $\formel0:=\# x\;\neg x=x$
  defines the number $0\in
  N(G)$.
The formula 
\[
\formel{succ}(y_1,y_2):=y_1\le y_2\wedge\neg y_1=y_2\wedge\forall
y(y\le y_1\vee y_2\le y)
\]
defines the
successor relation associated with the linear order $\le$.
  The following $\IFPC$-formula defines the set of even
  numbers in $N(G)$:
  \[
  \formel{even}(y):=\ifp\Big(Y\gets y\Bigmid y=\formel0\vee\exists y'\exists
  y''\big(Y(y')\wedge\formel{succ}(y',y'')\wedge\formel{succ}(y'',y)\big)\Big)y.
  \]
\end{exa}

\begin{exa}
  An \emph{Eulerian cycle} in a graph is a closed walk on which every edge
  occurs exactly once. A graph is \emph{Eulerian} if it has a Eulerian
  cycle. It is a well-known fact that a graph is Eulerian if and only if it is
  connected and every vertex has even degree.
  Then the following \IFPC-sentence defines the class of Eulerian graphs:
  \[
     \formel{eulerian}:=\formel{conn}\wedge\forall x_1\;\formel{even}\big(\# x_2\;E(x_1,x_2)\big),
  \]
  where $\formel{conn}$ is the sentence from Example~\ref{exa:conn} and
  $\formel{even}(y)$ is the formula from Example~\ref{exa:even}.
  By standard techniques from finite model theory, it can be proved that the
  class of Eulerian graphs is neither definable in \IFP\ nor in the counting
  extension \FOC\ of first-order logic.
\end{exa}

\subsection{Syntactical interpretations}

In the following, $\LL$ is one of the logics \IFPC, \IFP, or \FOL, and
$\lambda,\mu$ are relational languages such as the language $\{E\}$ of graphs
or the language $\{E,\les\}$ of ordered graphs. An
\emph{$\LL[\lambda]$-formula} is an $\LL$-formula in the language $\lambda$,
and similarly for $\mu$.  We need some
additional notation:
\begin{itemize}
\item  Let $\approx$ be an equivalence relation on a set
$U$. For every $u\in U$, by $u/_\approx$ we denote the $\approx$-equivalence
class of $u$, and we let $U/_\approx:=\{u/_\approx\mid u\in U\}$ be the set of
all equivalence classes. For a tuple $\vec u=(u_1,\ldots,u_k)\in U^k$ we let
$\vec u/_\approx:=(u_1/_\approx,\ldots,u_k/_\approx)$, and for a relation
$R\subseteq U^k$ we let $R/_\approx:=\{\vec u/_\approx\mid\vec u\in
R\}$. 
\item Two tuples $\bar x=(x_1,\ldots,x_k),(y_1,\ldots,y_\ell)$ of individual
  variables have the same \emph{type} if $k=\ell$ and for all $i\in[k]$ either
  both $x_i$ and $y_i$ range over vertices or both $x_i$ and $y_i$ range over
  numbers. For every structure $G$, we let $G^{\vec x}$ be the set of all tuples
  $\vec a\in(V(G)\cup N(G))^k$ such that for all $i\in[k]$ we have $a_i\in
  V(G)$ if $x_i$ is a vertex variable and $a_i\in N(G)$ if $x_i$ is a number
  variable.
\end{itemize}

\begin{defn}\label{def:si}
  \begin{enumerate}
  \item An \emph{\LL-interpretation of $\mu$ in $\lambda$} is a
    tuple
    \[
    \Gamma(\vec x)=\Big(\gapp(\vec x),\gamma_V(\vec
    x,\vec y),\gamma_\approx(\vec x,\vec y_1,\vec y_2),\big(\gamma_R(\vec
    x,\vec y_R)\big)_{R\in\mu}\Big),
    \]
    of $\LL[\lambda]$-formulae, where $\vec x$, $\vec y$, $\vec y_1$, $\vec
    y_2$, and
    $\vec y_R$ for $R\in\mu$ are tuples of individual variables such that
    $\vec y,\vec y_1,\vec y_2$ all have the same type, and for every $k$-ary
    $R\in\mu$ the tuple $\vec y_R$ can be written as $\vec y_{R1}\ldots\vec
    y_{R,k}$, where the $\vec y_{R,i}$ have the same type as $\vec y$.
  \end{enumerate}

  \medskip\noindent
  In the following, let $\Gamma(\vec
  x)$ be an
  \LL-interpretation of $\mu$ in $\lambda$. Let $G$ be a $\lambda$-structure
  and $\vec a\in G^{\vec x}$:
  \begin{enumerate}
    \setcounter{enumi}{2}
  \item $\Gamma(\vec x)$ is \emph{applicable} to $(G,\vec a)$ if
    $G\models\gapp[\vec a]$.
  \item If $\Gamma(\vec x)$ is applicable to $(G,\vec a)$, we let
    $\Gamma[G;\vec a]$ be the $\mu$-structure with vertex set 
    \[
    V\big(\Gamma[G;\vec a]\big):=\big\{\vec b\in G^{\vec y}\bigmid G\models\gamma_V[\vec a,\vec b]\big\}\big/_\approx,
    \]
    where $\approx$ is the reflexive, symmetric, transitive closure of
    $\big\{(\vec b_1,\vec b_2)\in (G^{\vec y})^2\bigmid G\models\gamma_\approx[\vec a,\vec
    b_1,\vec b_2]\big\}$. Furthermore, for $k$-ary $R\in\mu$, we let
    \[
    R\big(\Gamma[G;\vec a]\big):=\Big\{(\vec b_1,\ldots,\vec b_k)\in
    V\big(\Gamma[G;\vec a]\big)\Bigmid G\models\gamma_R[\vec a,\vec
    b_1,\ldots,\vec b_k]\Big\}\Big/_{\approx}.
    \] 
  \end{enumerate}
\end{defn}

Syntactical interpretations map $\lambda$-structures to $\mu$-structures.
The crucial observation is that they also induce a reverse translation
from $\LL[\mu]$-formulae to $\LL[\lambda]$-formulae. 

\begin{fact}[Lemma on Syntactical Interpretations]\label{fact:si}
  Let $\Gamma(\vec x)$ be an $\LL$-interpretation of $\mu$ in
  $\lambda$.  Then for every $\LL[\mu]$-sentence
  $\phi$ there is an $\LL[\lambda]$-formula
  $\phi^{-\Gamma}(\vec x)$ such
  that the following holds for all $\lambda$-structures $G$ and all tuples
  $\vec a\in G^{\vec x}$: If\/
  $\Gamma(\vec x)$ is applicable to $(G,\vec a)$, then
   \[
   G\models\phi^{-\Gamma}[\vec a]\iff\Gamma[G;\vec a]\models\phi.
   \]
\end{fact}
\noindent
A proof of this fact for first-order logic can be
found in \cite{ebbflutho94}. The proof for the other logics considered
here is an easy adaptation of the one for first-order logic.

\subsection{Definable canonisation}

A \emph{canonisation mapping} for a class of $\CC$ graphs associates with
every graph $G\in\CC$ an \emph{ordered copy} of $G$, that is, an ordered graph $(H,\le)$ such that $H\cong G$. We are
interested in canonisation mappings definable in the logic \IFPC\ by
syntactical interpretations of $\{E,\les\}$ in $\{E\}$. The easiest way to
define a canonisation mapping is by defining a linear order $\le$ on the universe of a
structure $G$ and then take $(G,\le)$ as the canonical copy. However, defining
an ordered copy of a structure is not the same as defining a linear order on
the universe, as
the following example illustrates:

\begin{exa}\label{exa:complete}
  Let $\CK$ be the class of all complete graphs. It is easy to see that there
  is no \IFPC-formula $\phi(x_1,x_2)$ such that for all $K\in\CK$ the binary
  relation $\phi[K;x_1,x_2]$ is a linear order of $V(K)$.

  However, there is an \FOC-definable canonisation
  mapping for the class $\CK$: Let \[\Gamma=\big(\gapp,\gamma_V(\vec
  y),\gamma_\approx(y_1,y_2),\gamma_E(y_1,y_2),\gamma_{\les}(y_1,y_2)\big)\] be the numerical
  $\FOC$-interpretation of $\{E,\les\}$ in $\{E\}$ defined by:
  \begin{itemize}
  \item $\gapp:=\forall x\; x=x$;
  \item $\gamma_V(y):= 1\le y\wedge y\le\formel{ord}$, where
    $\formel{ord}:=\#x\;x=x$;
  \item $\gamma_\approx(y_1,y_2):= y_1=y_2$;
  \item $\gamma_E(y_1,y_2):= \neg y_1=y_2$;
  \item $\gamma_\les(y_1,y_2):= y_1\le y_2$.
  \end{itemize}
  It is easy to see that the mapping $K\mapsto\Gamma[K]$ is a canonisation
  mapping for the class $\CK$.
\end{exa}

Our notion of \emph{definable canonisation} slightly relaxes the requirement
of defining a canonisation mapping; instead of just one ordered copy, we
associate with each structure a parametrised family of polynomially many
ordered copies.

\begin{defn}
  \begin{enumerate}
  \item   Let $\Gamma(\vec x)$ be an \LL-interpretation of $\{E,\les\}$ in $\{E\}$. 
Then $\Gamma(\vec x)$ \emph{canonises} a graph $G$ if there is
    at least one tuple $\vec a\in G^{\vec x}$ such that
    $\Gamma(\vec x)$ is applicable to $(G,\vec a)$, and for all tuples
    $\vec a\in G^{\vec x}$ such that $\Gamma(\vec x)$ is
    applicable to $(G,\vec a)$ it holds that $\Gamma[G;\vec a]$ is an ordered
    copy of $G$.
  \item A class $\CC$ of graphs admits \emph{\LL-definable canonisation} if there is an
    \LL-interpretation $\Gamma(\vec x)$ of $\{E,\les\}$ in $\{E\}$ that
    canonises all $G\in\CC$.
  \end{enumerate}
\end{defn}

The following well-known fact is a consequence of the Immerman-Vardi
Theorem. It is used, at least implicitly, in
\cite{gro98a,gro08a,gromar99,immlan90,ott97a}:

\begin{fact}
  Let $\CC$ be a class of graphs that admits \IFPC-definable canonisation. Then \IFPC\
  captures \PTIME\ on $\CC$.
\end{fact}

\section{Negative results}

In this section, we prove that \IFPC\ does not capture \PTIME\ on the
classes of chordal graphs and line graphs. Actually, our proof yields a more general result: Any logic
that captures \PTIME\ on any of these two classes and that is ``closed
under first-order reductions'' captures \PTIME\ on the class of all
graphs. It will be obvious what we mean by ``closed
under first-order reductions'' from the proofs, and it is also clear
that most ``natural'' logics will satisfy this closure condition. It
follows from our constructions that if there is a logic
capturing \PTIME\ on one of the two classes, then there is a logic
capturing \PTIME\ on all graphs.

Our negative results for \IFPC\ are based on the following theorem:

\begin{fact}[Cai, Fürer, and Immerman~\cite{caifurimm92}]\label{fact:cfi}
  There is a $\PTIME$-decidable property $\CP_{\textup{CFI}}$ of
  graphs that is not definable in \IFPC.
\end{fact}

\noindent
Without loss of generality we assume that all $G\in\CP_{\textup{CFI}}$ are
  connected and of order at least $4$.

\subsection{Chordal graphs}

Let us denote the class of chordal graphs by $\CRD$.

For every graph $G$, we define a graph $\hat G$ as follows:
\begin{itemize}
\item $V(\hat G):=V(G)\cup \{v_e\mid e\in E(G)\}$, where for each
  $e\in E(G)$ we let $v_e$ be a new vertex;
\item $\displaystyle E(\hat
  G):=\binom{V(G)}{2}\cup\big\{\{v,v_e\}\bigmid v\in V(G),e\in
  E(G),v\in e\big\}.$
\end{itemize}
The following lemmas collect the properties of the transformation
$G\mapsto\hat G$ that we need here. We leave the straightforward proofs
to the reader.

\begin{lem}\label{lem:crd1}
  For every graph $G$ the graph $\hat G$ is chordal.
\end{lem}

Note that for the graphs $K_2$ and $I_3:=\big([3],\emptyset\big)$ it
holds that $\hat K_2\cong\hat I_3\cong K_3$. It turns out that $K_2$
and $I_3$ are the only two nonisomorphic graphs that have isomorphic images under the
mapping $G\mapsto\hat G$. It is easy to verify this by observing that
for $G$ with $|G|\ge 4$ and $v\in V(\hat G)$, it holds that $v\in
V(G)$ if and only if $\deg(v)\ge 3$. Let $\hat\CG$ be the class of all
graphs $H$ such that $H\cong \hat G$ for some graph $G$.

\begin{lem}\label{lem:crd2}
  The class $\hat\CG$ is polynomial time decidable. Furthermore, there
  is a polynomial time algorithm that, given a graph $H\in\hat\CG$,
  computes the unique (up to isomorphism) graph
  $G\in\CG\setminus\{K\mid K\cong K_2\}$ with $\hat G\cong H$.
\end{lem}

\begin{lem}\label{lem:crd3}
  There is an \FOL-interpretation $\hat\Gamma$ of $\{E\}$ in $\{E\}$
  such that for all graphs $G$ it holds that $\hat\Gamma[G]\cong\hat
  G$.
\end{lem}

\begin{theo}\label{theo:neg-chordal}
  $\IFPC$ does not capture \PTIME\ on the class $\CRD$ of chordal graphs.
\end{theo}

\begin{proof}
  Let $\CFI$ be the graph property of Fact~\ref{fact:cfi} that
  separates \PTIME\ from  \IFPC. Note that 
  $K_2\not\in\CFI$ by our assumption that all graphs in $\CFI$ have order at
  least $4$. By Lemma~\ref{lem:crd2}, the class
  $\hat\CP:=\{H\mid H\cong \hat G\text{ for some }G\in\CFI\}$ is a polynomial time decidable subclass of $\CRD$. 

  Suppose for contradiction that \IFPC\ captures polynomial time on
  $\CRD$. Then by (G.1)$_{\CRD}$ there is an \IFPC-sentence $\phi$ such
  that for all chordal graphs $G$ it holds that $G\models\phi\iff
  G\in\hat\CP$. We apply the Lemma on Syntactical Interpretations to
  $\phi$ and the interpretation $\hat\Gamma$ of Lemma~\ref{lem:crd3}
  and obtain an \IFPC-sentence $\phi^{-\hat\Gamma}$ such that for all graphs
  $G$ it holds that
  \[
  G\models\phi^{-\hat\Gamma}\iff \hat G\cong\hat\Gamma[G]\models\phi.
  \]
  Thus $\phi^{-\hat\Gamma}$ defines $\CFI$, which is a contradiction.
\end{proof}

\subsection{Line graphs}
Let $\CL$ denote the class of all line graphs, or more precisely, the
class of all graphs $L$ such that there is a graph $G$ with
$L\cong L(G)$. Observe that a triangle and a claw have the same
line graph, a triangle. Whitney~\cite{whi32} proved that for all
nonisomorphic connected graphs $G,H$ except the claw and triangle, the line
graphs of $G$ and $H$ are nonisomorphic. The following fact,
corresponding to Lemma~\ref{lem:crd2}, is essentially an algorithmic version
of Whitney's result:

\begin{fact}[Roussopoulos~\cite{rou73}]\label{fact:line2}
  The class $\CL$ is polynomial time decidable. Furthermore, there
  is a polynomial time algorithm that, given a connected graph $H\in\CL$,
  computes the unique (up to isomorphism) graph
  $G\in\CG\setminus\{K\mid K\cong K_3\}$ with $L(G)\cong H$.
\end{fact}

\begin{lem}\label{lem:line3}
  There is an \FOL-interpretation $\Lambda$ of $\{E\}$ in $\{E\}$
  such that for all graphs $G$ it holds that $\Lambda[G]\cong L(G)$.
\end{lem}

\begin{proof}
  We define
  $\Lambda:=\big(\lambda_{\text{app}},\lambda_V(y_1,y_2),\lambda_\approx(y_1,y_2,y_1',y_2'),\lambda_E(y_1,y_2,y_1',y_2')\big)$
  by:
  \begin{itemize}
  \item $\lambda_{\text{app}}:=\forall x\;x=x$;
  \item $\lambda_V(y_1,y_2):= E(y_1,y_2)$;
  \item $\lambda_\approx(y_1,y_2,y_1',y_2'):=(y_1=y_1'\wedge
    y_2=y_2')\vee(y_1=y_2'\wedge y_2=y_1')$; 
  \item $\lambda_E(y_1,y_2,y_1',y_2'):=(y_1=y_1'\wedge
    \neg y_2=y_2')\vee(y_2=y_2'\wedge \neg y_1=y_1')\vee(y_1=y_2'\wedge \neg y_2=y_1')\vee(y_2=y_1'\wedge \neg y_2=y_1')$.
  \end{itemize}
\end{proof}

\begin{theo}\label{theo:neg-line}
  $\IFPC$ does not capture \PTIME\ on the class $\CL$ of line graphs.
\end{theo}

\begin{proof}
  The proof is completely analogous to the proof of
  Theorem~\ref{theo:neg-chordal}, using Fact~\ref{fact:line2} and
  Lemma~\ref{lem:line3} instead of Lemmas~\ref{lem:crd2} and \ref{lem:crd3}.
\end{proof}

\section{Capturing polynomial time on chordal line graphs}

In this section, we shall prove that \IFPC\ captures \PTIME\ on the
class $\CRD\cap\CL$ of graphs that are both chordal and line
graphs. As we will see, such graphs have a simple treelike
structure. We can exploit this structure and canonise the graphs in
$\CRD\cap\CL$ in a similar way as trees or graphs of bounded tree width.

\begin{exa}
  Figure~\ref{fig:cl} shows an example of a chordal line graph.
\end{exa}

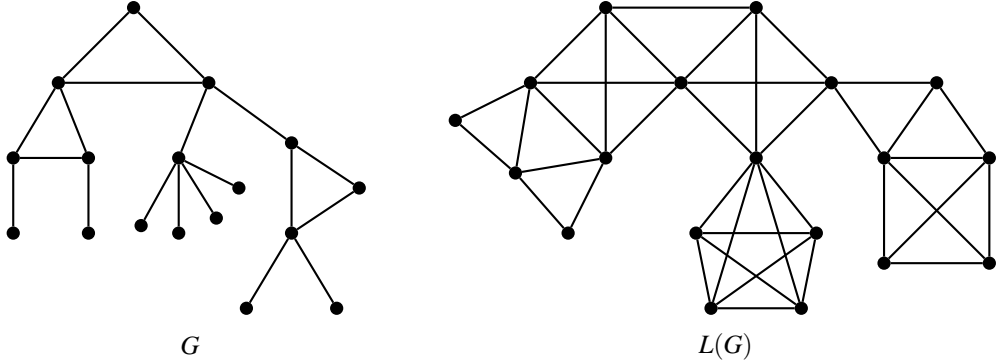
\begin{figure}
  \centering
  \begin{pspicture}(0.8,-0.5)(5.5,4)
    \dotnode(2.4,4){1}
    \dotnode(1.4,3){2}
    \dotnode(3.4,3){3}
    \dotnode(0.8,2){4}
    \dotnode(1.8,2){5}
    \dotnode(3,2){6}
    \dotnode(4.5,2.2){7}
    \dotnode(0.8,1){8}
    \dotnode(1.8,1){9}
    \dotnode(2.5,1.1){10}
    \dotnode(3,1){11}
    \dotnode(3.5,1.2){12}
    \dotnode(3.8,1.6){13}
    \dotnode(5.4,1.6){14}
    \dotnode(4.5,1){15}
    \dotnode(3.9,0){16}
    \dotnode(5.1,0){17}
    
    \ncline{1}{2}
    \ncline{1}{3}
    \ncline{2}{3}
    \ncline{2}{4}
    \ncline{2}{5}
    \ncline{3}{6}
    \ncline{3}{7}
    \ncline{4}{5}
    \ncline{4}{8}
    \ncline{5}{9}
    \ncline{6}{10}
    \ncline{6}{11}
    \ncline{6}{12}
    \ncline{6}{13}
    \ncline{7}{15}
    \ncline{7}{14}
    \ncline{14}{15}
    \ncline{15}{16}
    \ncline{15}{17}

    \rput(3.15,-0.5){$G$}
  \end{pspicture}
  \hspace{1cm}
  \begin{pspicture}(1,-0.5)(8.2,4)
    \dotnode(3,4){a}
    \dotnode(5,4){b}
    \dotnode(4,3){c}
    \dotnode(2,3){d}
    \dotnode(3,2){e}
    \dotnode(5,2){f}
    \dotnode(6,3){g}
    \dotnode(1.8,1.8){h}
    \dotnode(1,2.5){i}
    \dotnode(2.5,1){j}
    \dotnode(4.2,1){k}
    \dotnode(5.8,1){l}
    \dotnode(4.4,0){m}
    \dotnode(5.6,0){n}
    \dotnode(6.7,2){o}
    \dotnode(7.4,3){p}
    \dotnode(8.1,2){q}
    \dotnode(6.7,0.6){r}
    \dotnode(8.1,0.6){s}

    \ncline{a}{b}
    \ncline{a}{c}
    \ncline{a}{d}
    \ncline{a}{e}
    \ncline{b}{c}
    \ncline{b}{f}
    \ncline{b}{g}
    \ncline{c}{d}
    \ncline{c}{e}
    \ncline{c}{f}
    \ncline{c}{g}
    \ncline{d}{e}
    \ncline{d}{h}
    \ncline{d}{i}
    \ncline{e}{h}
    \ncline{e}{j}
    \ncline{f}{g}
    \ncline{f}{k}
    \ncline{f}{l}
    \ncline{f}{m}
    \ncline{f}{n}
    \ncline{g}{p}
    \ncline{g}{o}
    \ncline{h}{i}
    \ncline{h}{j}
    \ncline{k}{l}
    \ncline{k}{m}
    \ncline{k}{n}
    \ncline{l}{m}
    \ncline{l}{n}
    \ncline{m}{n}
    \ncline{o}{p}
    \ncline{o}{q}
    \ncline{o}{r}
    \ncline{o}{s}
    \ncline{p}{q}
    \ncline{q}{r}
    \ncline{q}{s}
    \ncline{r}{s}

    \rput(4.6,-0.5){$L(G)$}
  \end{pspicture}

  \caption{A graph $G$ and its line graph $L(G)$, which is chordal}
  \label{fig:cl}
\end{figure}

\subsection{On the structure of chordal line graphs}

It is a well-known fact that chordal graphs can be decomposed into cliques
arranged in a tree-like manner. To state this formally, we review tree
decompositions of graphs. A \emph{tree decomposition} of a
graph $G$ is a pair $(T,\beta)$, where $T$ is a tree and $\bag:V(T)\to
2^{V(G)}$ is a mapping such that the following two conditions are
satisfied:
\begin{nlist}{T}
\item\label{li:t1} For every $v\in V(G)$ the set $\{t\in
  V(T)\mid v\in\bag(t)\}$ is connected in $T$.
\item\label{li:t2} For every $e\in E(G)$ there is a $t\in
  V(T)$ such that $e\subseteq\bag(t)$.
\end{nlist}
The sets $\bag(t)$, for $t\in V(T)$, are called the \emph{bags} of the
decomposition. It will be convenient for us to always assume the tree $T$
in a tree decomposition to be rooted. This gives us the partial tree order
$\dagle^T$.
We introduce some additional notation. Let $(T,\beta)$ be a tree decomposition
of a graph $G$. For every $t\in V(T)$ we let:
\begin{align*}
  \cone(t)&:=\bigcup_{u\in V(T)\text{ with }t\dagle^T u}\bag(u),
\end{align*}
The set $\cone(t)$ is called the \emph{cone} of $(T,\beta)$ at $t$. It
easy to see that for every $t\in V(T)\setminus\{r(T)\}$ with parent $s$ the set $\bag(t)\cap\bag(s)$ separates
$\cone(t)$ from $V(G)\setminus\cone(t)$. Furthermore, for every clique
$X$ of $G$ there is a $t\in V(T)$ such that $X\subseteq\bag
(t)$. (See
Diestel's textbook \cite{die05} for proofs of these facts and
background on tree decompositions.) Another useful fact is that every tree decomposition $(T,\beta)$
of a graph $G$ can be transformed into a tree decomposition
$(T',\beta')$ such that for all $t'\in V(T')$ there exists a $t\in
V(T)$ such that $\beta'(t')=\beta(t)$, and for all $t,u\in V(T')$ with
$t\neq u$ it holds that $\beta'(t)\not\subseteq \beta'(u)$. 

\begin{fact}\label{fact:crd-cli}
  A nonempty graph $G$ is chordal if and only if $G$ has a tree decomposition into
  cliques, that is, a tree decomposition $(T,\beta)$ such that for all $t\in
  V(T)$ the bag $\beta(t)$ is a clique of $G$.
\end{fact}

For a graph $G$, we let $\MCL(G)$ be the set of all maximal cliques in $G$ with
respect to set inclusion. 
If we combine Fact~\ref{fact:crd-cli} with the observations about
tree decomposition stated before the fact, we obtain the following lemma:

\begin{lem}\label{lem:good-dec}
  Let $G$ be a nonempty chordal graph. Then $G$ has a tree decomposition $(T,\beta)$
  with the following properties:
  \begin{eroman}
  \item For every $t\in V(T)$ it holds that $\bag(t)\in\MCL(G)$.
  \item For every $X\in \MCL(G)$ there is exactly one $t\in
    V(T)$ such that $\bag(t)=X$.
  \end{eroman}
  We call a tree decomposition satisfying conditions (i) and (ii) a \emph{good
    tree decomposition} of $G$.
\end{lem}

Let us now turn to line graphs. Let $L:=L(G)$ be the line graph of a graph $G$. For
every $v\in V(G)$, let $X(v):=\{e\in E(G)\mid v\in e\}\subseteq
V(L)$. Unless $v$ is an isolated vertex, $X(v)$ is a clique in $L$. Furthermore, we have
\[
 L=\bigcup_{v\in V(G)}L[X(v)].
\]
Observe that for all $v,w\in V(G)$, if $e:=\{v,w\}\in E(G)$ then $X(v)\cap
  X(w)=\{e\}$, and if $\{v,w\}\not\in E(G)$ then $X(v)\cap
  X(w)=\emptyset$.
The following proposition, which is probably well-known, characterises
the line graphs that are chordal:

\begin{prop}\label{prop:lc}
  Let $L=L(G)\in\CL$. Then 
  \[
  L\in\CRD\iff\text{all cycles in $G$ are triangles}.
  \]
\end{prop}

  Note that on the right hand side, we do not only consider chordless cycles.

\begin{proof}
  For the forward direction, suppose that $L\in\CRD$, and let $C\subseteq
  G$ be a cycle. Then $L[E(C)]$ is a chordless
  cycle in $L$. Hence $|C|\le 3$, that is, $C$ is a triangle.

  For the backward direction, suppose that all cycles in $G$ are
  triangles, and let $C\subseteq L$ be a chordless cycle of length
  $k$. Let $e_1,\ldots,e_k$ be the vertices of $C$ in cyclic
  order. To simplify the notation, let $e_0:=e_k$. Then for all
  $i\in[k]$ it holds that $\{e_{i-1},e_i\}\in E(L)$ and thus
  $e_{i-1}\cap e_i\neq\emptyset$. Let $v_0,v_1\in V(G)$ such that
  $e_1=\{v_0,v_1\}$, and for $i\in[2,k]$, let $v_i\in e_i\setminus
  e_{i-1}$. Then $v_i\neq v_j$ for all $j\in[i-2]$, and if $i<k$ even for $j\in[0,i-2]$, because the cycle $C$ is
  chordless and thus $e_i\cap e_j=\emptyset$. Furthermore,
  $v_k=v_0$. Thus $\{v_1,\ldots,v_k\}$ is the vertex set of a cycle in
  $G$, and we have $k= 3$.
\end{proof}

\begin{lem}\label{lem:lc1}
  Let $L=L(G)\in\CRD\cap\CL$, and let
  $X\in\MCL(L)$ and $e=\{v,w\}\in X$. Then $X=X(v)$ or $X=X(w)$ or
    there is an $x\in V(G)$ such that $\{x,v\},\{x,w\}\in E(G)$ and
    $X=\big\{e,\{x,v\},\{x,w\}\big\}$. 
\end{lem}  

\begin{proof}
  For all $f\in X$, either $v\in f$ or $w\in f$, because $f$ is
  adjacent to $e$. Hence $X\subseteq X(v)\cup X(w)$. If
  $X\subseteq X(v)$, then $X=X(v)$ by the maximality of
  $X$. Similarly, if $X\subseteq X(w)$ then $X=X(w)$. Suppose
  that $X\setminus X(v)\neq\emptyset$ and $X\setminus
  X(w)\neq\emptyset$. Let $f\in X\setminus X(v)$ and $g\in
  X\setminus X(w)$. As $X$ is a clique, we have $\{f,g\}\in E(L)$
  and thus $f\cap g\neq\emptyset$. Hence there is an $x\in V(G)$ such
  that $f=\{x,w\}$ and $g=\{x,v\}$. Furthermore, $X=\{e,f,g\}$. To
  see this, let $h\in X$. Then $\{h,e\}\in E(L)$ and thus $v\in
  h$ or $w\in h$. Say, $v\in h$. If $w\in h$, then $h=e$. Otherwise, we have
  $x\in h$, because $h$ is adjacent to $g$. Thus $h=g$.
\end{proof}

\begin{lem}\label{lem:lc2}
    Let $L\in\CRD\cap\CL$, and let
  $X_1,X_2\in\MCL(L)$ be distinct. Then $|X_1\cap X_2|\le 2$.
\end{lem}

\begin{proof}
    Let $L=L(G)$ for some graph $G$. Suppose for contradiction that
    $|X_1\cap X_2|\ge 3$. Then $|X_1|,|X_2|\ge 4$, because $X_1$ and $X_2$ are
    distinct maximal cliques. By Lemma~\ref{lem:lc1}, it follows that there
    are vertices $v_1,v_2\in V(G)$ such that $X_1=X(v_1)$ and
    $X_2=X(v_2)$, which implies $|X_1\cap X_2|\le 1$. This is a contradiction.
\end{proof}

\begin{lem}\label{lem:lc3}
  Let $L\in\CRD\cap\CL$, and let
  $X_1,X_2,X_3\in\MCL(L)$ be pairwise distinct such that $X_1\cap
  X_2\cap X_3\neq\emptyset$. Then there are $i,j,k$ such that
  $\{i,j,k\}=[3]$ and $X_i\subseteq X_j\cup X_k$ and $|X_i|=3$.
\end{lem}

\begin{proof}
  Let $L=L(G)$ for some graph $G$.
  Let $e\in X_1\cap X_2\cap X_3$. Suppose that $e=\{v,w\}\in
  E(G)$. 
  As the cliques $X_1,X_2,X_3$ are distinct, it follows from Lemma~\ref{lem:lc1}
  that there is an $i\in[3]$ and an $x\in V(G)$ such that
  $X_i=\big\{e,\{x,v\},\{x,w\}\big\}$. Choose such $i$ and $x$.

  \begin{claim}1
    For all $j\in[3]\setminus\{i\}$, either $X_j=X(v)$ or $X_j=X(w)$.

    \proof
    Suppose for contradiction that $X_j\neq X(v)$ and $X_j\neq
    X(w)$. Then by Lemma~\ref{lem:lc1}, there exists a $y\in V(G)$ such that $\{y,v\},\{y,w\}\in E(G)$ and
    $X_j=\big\{e,\{y,v\},\{y,w\}\big\}$. But then 
    \[
    L\big[\{y,v\},\{v,x\},\{x,w\},\{w,y\}\big]
    \]
    is a chordless cycle in $L$, which contradicts $L$ being chordal.
    \uend
  \end{claim}

  Thus there are $j,k$ such that $\{i,j,k\}=[3]$ and $X_j=X(v)$ and
  $X_k=X(w)$. Then $X_i\subseteq X_j\cup X_k$.
\end{proof}

\begin{lem}\label{lem:lc4}
  Let $L\in\CRD\cap\CL$. Then every good tree decomposition $(T,\beta)$ of $L$
  satisfies the following conditions (in addition to conditions (i) and
  (ii) of Lemma~\ref{lem:good-dec}):
  \begin{eroman}
    \setcounter{erom}{2}
  \item For all $t\in V(T)$, 
    \begin{itemize}
    \item either $|\bag(t)|=3$ and $t$ has at most three neighbours in
      $T$ (the \emph{neighbours} of a node are its children and the
      parent),
    \item or for all distinct neighbours $u,u'$ of\/ $t$\/ in $T$ it holds
      that $\bag(u)\cap\bag(u')=\emptyset$.
    \end{itemize}
  \item For all $t,u\in V(T)$ with $t\neq u$ it holds that
    $|\bag(t)\cap\bag(u)|\le 2$.
  \end{eroman}
\end{lem}

\begin{proof}
  Let $(T,\beta)$ be a good tree decomposition of $L$. Such a
  decomposition exists because $L$ is chordal. As all bags of the
  decomposition are maximal cliques of $L$, condition (iii) follows
  from Lemma~\ref{lem:lc3} and condition (iv) follows from
  Lemma~\ref{lem:lc2}.
\end{proof}

\subsection{Canonisation}

\begin{theo}\label{theo:can}
  The class $\CRD\cap\CL$ of all chordal line graphs admits \IFPC-definable
  canonisation.
\end{theo}

\begin{cor}
  \IFPC\ captures \PTIME\ on the class of all chordal line graphs.
\end{cor}

\begin{proof}[Proof of Theorem~\ref{theo:can}]
  The proof resembles the proof that classes of graphs of bounded tree
  width admit \IFPC-definable canonisation \cite{gromar99} and also
  the proof of Theorem~7.2 (the ``Second Lifting Theorem'') in
  \cite{gro08a}. Both of these proofs are generalisations of the
  simple proof that the class of trees admits \IFPC-definable
  canonisation (see, for example, \cite{gklmsvvw07-3}). We shall
  describe an inductive construction that associates with each
  chordal line graph $G$ a canonical copy $G'$ whose universe is
  an initial segment of the natural numbers. For readers with some
  experience in finite model theory, it will be straightforward to
  formalise the construction in \IFPC.  We only describe the
  canonisation of \emph{connected} chordal line graphs that are not
  complete graphs. It is
  easy to extend it to arbitrary chordal line graphs. For complete
  graphs, which are chordal line graphs, cf.~Example~\ref{exa:complete}

  To describe the construction, we fix a connected graph
  $G\in\CRD\cap\CL$ that is not a complete graph. Note that this
  implies $|G|\ge 3$. Let $(T,\beta^T)$ be a good tree
  decomposition of $G$. As $G$ is not a complete graph, we have
  $|T|\ge 2$. Without loss of generality we may assume that
  the root $r(T)$ has exactly one child in $T$, because every tree has
  at least one node of degree at most $1$ and properties (i), (ii) of
  a good decomposition do not depend on the choice of the root. It
  will be convenient to view the rooted tree $T$ as a directed graph,
  where the edges are directed from parents to children.

  Let $U$ be the set of all triples $(u_1,u_2,u_3)\in V(G)^3$ such that $u_3\neq
  u_1,u_2$ (possibly, $u_1=u_2$), and there is a unique $X\in\MCL(G)$
  such that $u_1,u_2,u_3\in X$. For all $\vec u=(u_1,u_2,u_3)\in U$, let $A(\vec u)$ be the connected component of
  $G\setminus\{u_1,u_2\}$ that contains $u_3$ (possibly, $A(\vec u)=G\setminus\{u_1,u_2\}$). We define mappings $\sep^U,\comp^U,\cone^U,\bag^U:U\to 2^{V(G)}$
  as follows: For all $\vec u=(u_1,u_2,u_3)\in U$, we let $\sep^U(\vec
      u):=\{u_1,u_2\}$ and $\comp^U(\vec u):=V(A(\vec u))$. We let
      $\cone^U(\vec u):=\sep^U(\vec u)\cup\comp^U(\vec u)$, and we let
      $\bag^U(\vec u)$ the unique $X\in\MCL(G)$ with $u_1,u_2,u_3\in X$.
  We define a partial order $\dagle$ on $U$ by letting $\vec
  u\dagle\vec v$ if and only if $\vec u=\vec v$ or $\comp(\vec
  u)\supset\comp(\vec v)$.  We let $F$ be the successor relation of
  $\dagle$, that is, $(\vec u,\vec v)\in F$ if $\vec u\dagsle\vec v$
  and there is no $\vec w\in U\setminus\{\vec u,\vec v\}$ such that
  $\vec u\dagsle\vec w\dagsle\vec v$. Finally, we let $D:=(U,F)$. Then
  $D$ is a directed acyclic graph. It is easy to verify that for all 
  $\vec u\in U$ we have
  \begin{equation}\label{eq:bag}
  \bag^U(\vec u)=\cone^U(\vec u)\setminus\bigcup_{\vec v\in N^D(\vec
    u)}\comp^U(\vec v),
  \end{equation}
  where $N^D(\vec u)=\big\{\vec v\in U\bigmid (\vec u,\vec v)\in F\big\}$.

  \medskip
  Recall that we also have mappings $\bag^T,\cone^T:V(T)\to 2^{V(G)}$
  derived from the tree decomposition. We define a mapping
  $\sep^T:V(T)\to 2^{V(G)}$ as follows:
  \begin{itemize}
  \item For a node $t\in V(T)\setminus\{r(T)\}$ with parent $s$, we
    let $\sep^T(t):=\bag^T(t)\cap\bag^T(s)$.
  \item For the root $r:=r(T)$, we first define a set $S\subseteq V(G)$
    by letting 
    $S:=\bag^T(r)\setminus\bag^T(t)$, where $t$ is the unique child of
    $r$. (Remember our assumption that $r$ has exactly one child.)
    Then if $|S|\ge 2$, we choose distinct $v,v'\in S$ and let
    $\sep^T(r):=\{v,v'\}$, and if $|S|=1$ we let $\sep^T(r):=S$.
   \end{itemize}
   Note that $\bag^T(t)\setminus\sep^T(t)\neq\emptyset$ and
   $1\le|\sep^T(t)|\le 2$ for all $t\in V(T)$. For the root, this
   follows immediately from the definition of $\sep^T(t)$, and for nodes $t\in
   V(T)\setminus\{r(T)\}$ it follows from Lemma~\ref{lem:lc4}.
   We define a mapping $\comp^T:V(T)\to 2^{V(G)}$ by letting
   $\comp^T(t):=\cone^T(t)\setminus\sep^T(t)$ for all $t\in V(T)$.  We
   define a mapping $g:V(T)\to U$ by choosing, for every node $t\in
   V(T)$, vertices $u_1,u_2$ such that $\sep^T(t)=\{u_1,u_2\}$
   (possibly $u_1=u_2$) and a
   vertex $u_3\in\bag(t)\setminus\sep(t)$ and letting
   $g(t):=(u_1,u_2,u_3)$. Note that $(u_1,u_2,u_3)\in U$, because
   $\bag^T(t)$ is the unique maximal clique in $\MCL(G)$ that contains $u_1,u_2,u_3$.

   \begin{claim}1
     The mapping $g$ is a
     directed graph embedding of $T$ into $D$. Furthermore, for all $t\in
     V(T)$ it holds that $\comp^T(t)=\comp^U(g(t))$, $\bag^T(t)=\bag^U(g(t))$,
     $\cone^T(t)=\cone^U(g(t))$, and $\sep^T(t)=\sep^U(g(t))$.
     
     \proof We leave the straightforward inductive proof to the
     reader.
     \uend
   \end{claim}

  Let $\vec u_0:=g(r(T))$, and let $U_0$ be the subset of $U$ consisting of all $\vec u\in U$ such
  that $\vec u_0\dagle\vec u$. Let $F_0$ be the restriction of $F$
  to $U_0$ and $D_0:=(U_0,F_0)$. Note that $U_0$ is upward closed with
  respect to $\dagle$ and that $g(T)\subseteq D_0$. 

   \begin{claim}2
     There is a mapping $h:U_0\to V(T)$ such that $h$ is a directed graph
     homomorphism from $D_0$ to $T$ and $h\circ g$ is the identity
     mapping on $V(T)$. Furthermore, for all $\vec u\in U_0$ it holds
     that $\comp^U(\vec u)=\comp^T(h(\vec u))$, $\bag^U(\vec u)=\bag^T(h(\vec u))$,
     $\cone^U(\vec u)=\cone^T(h(\vec u))$, and $\sep^U(\vec u)=\sep^T(h(\vec u))$.
     
     \proof We define $h$ by induction on the partial order $\dagle$.
     The unique $\dagle$-minimal element of $U_0$ is $\vec u_0$. We
     let $h(\vec u_0):=r(T)$. Now let $\vec v=(v_1,v_2,v_3)\in U_0$,
     and suppose that $h(\vec u)$ is defined for all $\vec u\in U_0$
     with $\vec u\dagsle\vec v$. Let $\vec u\in U_0$ such that $(\vec
     u,\vec v)\in F_0$, and let $s:=h(\vec u)$. By the induction
     hypothesis, we have $\comp^U(\vec u)=\comp^T(s)$, $\bag^U(\vec
     u)=\bag^T(s)$, $\cone^U(\vec u)=\cone^T(s)$, and $\sep^U(\vec
     v)=\sep^T(s)$. The set $\comp^U(\vec v)$ is the vertex set of a
     connected component of $G\setminus\sep^U(\vec v)$ which is
     contained in $\comp^U(\vec u)\subseteq\cone^U(\vec
     u)=\cone^T(s)$, and by \eqref{eq:bag} it holds that $\comp^U(\vec
     v)\cap\bag^U(\vec u)=\emptyset$. Hence there is a child $t$ of $s$
     such that $\comp^U(\vec v)\subseteq\comp^T(t)$. Let $\vec
     v':=g(t)$. If $\comp^U(\vec v)\subset\comp^T(t)=\comp^U(\vec v')$, then
     $\vec u\dagsle\vec v'\dagsle\vec v$, which contradicts $(\vec
     u,\vec v)\in F$. Hence $\comp^U(\vec v)=\comp^T(t)$ and thus
     $\sep^U(\vec v)=\sep^T(t)$. This also implies $\cone^U(\vec
     v)=\cone^T(t)$ and $\bag^U(\vec v)=\bag^T(t)$. We let $h(\vec
     v):=t$.
 
     To prove that $h$ is really a homomorphism,
     it remains to prove that for all $\vec u'\in U_0$ with $(\vec
     u',\vec v)\in F_0$ we also have $h(\vec u')=s$. So let $\vec u'\in U_0$ with $(\vec
     u',\vec v)\in F_0$, and let $s'=h(\vec u')$. Suppose for
     contradiction that $s\neq s'$. If $s'\dagsle^T s$ then
     $\comp^U(\vec u')\supset\comp^U(\vec u)$ and thus $\vec u'\dagsle\vec
     u$, which contradicts $(\vec
     u',\vec v)\in F_0$. Thus $s'\not\dagle^T s$, and similarly
     $s\not\dagle^T s'$. But then both $\sep^T(s)$ and $\sep^T(s')$ separate
     $\cone^T(s)$ from $\cone^T(s')$ in $G$. This contradicts $\comp^U(\vec
     v)\subseteq\comp^T(s)\cap\comp^T(s')\subseteq\big(\cone^T(s)\cap\cone^T(s')\big)\setminus\big(\sep^T(s)\cup\sep^T(s')\big)$.
     \uend
   \end{claim}

  Thus essentially, the ``treelike'' decomposition $(D_0,\beta^U)$ is the
  same as the tree decomposition $(T,\beta^T)$. However, the
  decomposition $(D_0,\beta^U)$ is \IFP-definable with three parameters fixing
  the tuple $\vec u_0=g(r(T))$.

  Let us now turn to the canonisation. For every $\vec u\in U_0$, we
  let $G(\vec u):=G[\cone(\vec u)]$. Then $G=G({\vec u_0})$. We inductively define for every
  $\vec u=(u_1,u_2,u_3)\in U_0$ a graph $H({\vec u})$ with the following properties:
  \begin{eroman}
    \item
      $V\big(H({\vec u})\big)=[n_{\vec u}]$, where $n_{\vec
        u}:=|\cone(\vec u)|=\big|V\big(G_{\vec u})\big)\big|$.
    \item There is  an isomorphism $f_{\vec u}$ from $G({\vec u})$ to
      $H({\vec u})$ such that if $u_1\neq u_2$ it holds that $f_{\vec
        u}(u_1)=1$ and  $f_{\vec
        u}(u_2)=2$, and if $u_1= u_2$ it holds that $f_{\vec
        u}(u_1)=1$.
  \end{eroman}
  For the induction basis, let $\vec u\in U_0$ with $N^{D_0}(\vec
  u)=\emptyset$. Then $\cone^U(\vec
  u)=\bag^U(\vec u)$, and $G({\vec u})=K[\bag^U(\vec u)]$. We let
  $n:=n_{\vec u}=|\beta^U(\vec u)|$ and $H({\vec
    u}):=K_{n}$. Then (i) and (ii) are obviously satisfied.

  For the induction step, let $\vec u\in U_0$ and $N^{D_0}(\vec
  u)=\{\vec v^1,\ldots,\vec v^n\}\neq\emptyset$. It follows from
  Claim~2 that for all $i,j\in[n]$, either $\cone(\vec v^i)=\cone(\vec
  v^j)$ or $\cone(\vec v^i)\cap\cone(\vec v^j)=\sep(\vec
  v^i)\cap\sep(\vec v^j)\subseteq\bag(\vec u)$.  We may assume without
  loss of generality that there are $i_1,\ldots,i_m\in[n]$ such
  that $i_1<i_2<\ldots<i_m$ and for all $j,j'\in[m]$ with $j\neq j'$ we have $\cone(\vec v^{i_j})\neq\cone(\vec
  v^{i_{j'}})$ and for all $j\in[m]$,
  $i\in[i_j,i_{j+1}-1]$ we have $\cone(\vec v^i)=\cone(\vec
  v^{i_j})$. Here and in the
  following we let $i_{m+1}:=n+1$. 

  The class of all graphs whose vertex set is a subset of $\PN$ may
  be ordered lexicographically; we let $H\slex H'$ if either $V(H)$
  is lexicographically smaller than $V(H')$, that is, the first
  element of the symmetric difference $V(H)\triangle V(H')$ belongs
  to $V(H')$, or $V(H)=V(H')$ and $E(H)$ is lexicographically smaller
  than $E(H')$ with respect to the lexicographical ordering of
  unordered pairs of natural numbers, or $H=H'$.
  Without loss of
  generality we may assume that for each $j\in[m]$ it holds that
  \begin{equation*}
    \label{eq:lex1}
    H(\vec v^{i_j})\slex H(\vec v^{i_j+1})\slex H(\vec v^{i_j+2})\slex\ldots\slex H(\vec v^{i_{j+1}-1})
  \end{equation*}
  and, furthermore, 
  \begin{equation}
    \label{eq:lex2}
    H(\vec v^{i_1})\slex H(\vec v^{i_2})\slex \ldots\slex H(\vec v^{i_m})
  \end{equation}
  Note that, even though the graphs $G(\vec v^{i_1}),G(\vec
  v^{i_2}),\ldots,G(\vec v^{i_m})$ are vertex disjoint subgraphs of
  $G(\vec u)$, they may be isomorphic, and hence not all of the
  inequalities in \eqref{eq:lex2} need to be strict. For all $j\in [m]$,
  let $\vec v_j:=\vec v^{i_j}$ and $G_j:=G(\vec v_j)$ an $H_j:=H(\vec
    v_j)$. Then $H_1\slex H_2\slex\ldots\slex H_m$. Let
  $j_1,\ldots,j_\ell\in[m]$ such that $j_1<j_2<\ldots<j_\ell$ and
  $H_j=H_{j_i}$ for all $i\in[\ell]$, $j\in[j_i,j_{i+1}-1]$, where
  $j_{\ell+1}=m+1$, and $H_{j_i}\neq H_{j_{i+1}}$ for all
  $i\in[\ell-1]$. For all $i\in[\ell]$, let $J_i:=H_{j_i}$. Furthermore, let
  $n_i:=|J_i|$ and $k_i:=j_{i+1}-j_i$ and $q_i:=|\sep^U(\vec
  v^{i_j})|$ and
  \[
  q:=\left|\bag^U(\vec u)\setminus\bigcup_{j=1}^m\bag^U(\vec v_j)\right|.
  \]
  \begin{cs}
    \case1
    For all neighbours $t,t'$ of $h(\vec u)$ in the undirected tree underlying
    $T$ it holds that $\bag^T(t)\cap\bag^T(t')=\emptyset$.\\
  We define $H(\vec u)$ by first taking a complete graph $K_q$, then
  $k_1$ copies of $J_1$, then $k_2$ copies of $J_2$, et cetera, and
  finally $k_\ell$ copies of $J_\ell$. The universes of all these
  copies are disjoint, consecutive intervals of natural numbers. Let
  $K$ be the union of $[q]$ with the first $q_i$ vertices of each of the $k_i$ copies of
  $J_i$
  for all $i\in[\ell]$. Then $K$ is the set of vertices of $H(\vec u)$
  that corresponds to the clique $\bag(\vec u)$. We add edges among the vertices
  in $K$ to turn it into a clique. It is not hard to verify that the
  resulting structure satisfies (i) and (ii).

  \case2
  There are neighbours $t,t'$ of $h(\vec u)$ in the undirected tree underlying
    $T$ such that $\bag^T(t)\cap\bag^T(t')\neq\emptyset$.\\
    Then by Lemma~\ref{lem:lc4}(iii) we have $|\bag^U(\vec u)|=3$, and
    $h(\vec u)$ has at most two children. Hence $m\le 2$, and
    essentially this means we only have two possibilities of how to
    combine the parts $H_1,H_2$ to the graph $H(\vec u)$; either $H_1$
    comes first or $H_2$. We choose the lexicographically smaller possibility.
    We omit the details.
  \end{cs}
  This completes our description of the construction of the graphs
  $H({\vec u})$.

  It remains to prove that $H({\vec u})$ is \IFPC-definable. We first
  define \IFP-formulae $\theta_U(\vec x)$, $\theta_F(\vec x,\vec y)$,
  $\theta_{\comp}(\vec x,y)$, $\theta_{\bag}(\vec x,y)$, $\theta_{\cone}(\vec x,y)$,
  $\theta_{\sep}(\vec x,y)$ such that 
  \begin{align*}
    U&=\big\{\vec u\in V(G)^3\bigmid
  G\models\theta_U[\vec u]\big\},\\
    F&=\big\{(\vec u,\vec v)\in U^2\bigmid
  G\models\theta_F[\vec u,\vec v]\big\},\\
  \comp^U(\vec u)&=\big\{v\in V(G)\bigmid G\models\theta_\comp[\vec
  u,v]\big\}
  &\text{for all }\vec u\in U,
  \end{align*}
  and similarly for $\bag,\cone,\sep$.  Then we define formulae
  $\theta_U^0(\vec x_0,\vec x)$, $\theta_F^0(\vec x_0,\vec x)$ that
  define $D_0$. We have no canonical way of checking that a tuple
  $\vec u_0$ really is the image $g(r(T))$ of the root of a good tree
  decomposition, but all we need is that the graph $D^0(\vec u_0)$
  with vertex set $\big\{\vec u\in V(G)^3\bigmid
  G\models\theta_U^0[\vec u_0,\vec u]\big\}$ and edge set $\big\{(\vec
  u,\vec v)\in U^2\bigmid G\models\theta_F[\vec u_0,\vec u,\vec
  v]\big\}$ has the properties we derive from $T$ being a good tree
  decomposition. In particular, if a node $\vec u$ has a child $\vec
  v$ with $\sep^U(\vec u)\cap\sep^U(\vec v)\neq\emptyset$ or children
  $\vec v_1\neq\vec v_2$ with $\sep^U(\vec v_1)\cap\sep^U(\vec
  v_1)\neq\emptyset$, then $|\bag^U(\vec u)|\le 3$. Once we have defined
  $D^0$, it is straightforward to formalise the definition of the
  graphs $H({\vec u})$ in \IFPC\ and define an \IFPC-interpretation
  $\Gamma(\vec x_0)$ that canonises $G$. We leave the (tedious)
  details to the reader.
\end{proof}

\begin{rem}\label{rem:gen}
  Implicitly, the previous proof heavily depends on the concepts
  introduced in \cite{gro08a}. In particular, the definable directed
  graph $D$ together with the definable mappings $\sep$ and $\comp$
  constitute a \emph{definable tree decomposition}. However, our
  theorem does not follow directly from Theorem~7.2 of \cite{gro08a}.

  The class $\CRD\cap\CL$ of chordal line graphs is fairly restricted,
  and there may be an easier way to prove the canonisation theorem by using
  Proposition~\ref{prop:lc}. The proof given
  here
  has
  the advantage that it generalises to the class of all chordal graphs
  that have a good tree decomposition where the bags of the neighbours
  of a node intersect in a ``bounded way''. We omit the details.
\end{rem}

\section{Further research}
I mentioned several important open problems related to the quest for a
logic capturing \PTIME\ in the survey in Section~\ref{sec:survey}. Further open
problems can be found in~\cite{gro08b}. Here, I will briefly 
discuss a few open problems related to classes closed under taking
induced subgraphs, or equivalently, classes defined by excluding
(finitely or infinitely many) induced subgraphs. 

A fairly obvious, but not particularly interesting generalisation of our positive capturing result is pointed
out in Remark~\ref{rem:gen}. I conjecture that our theorem for
chordal line graphs
can be generalised to the class of chordal claw-free graphs, that is, I conjecture that the
class of chordal claw-free graphs admits \IFPC-definable canonisation. Further
natural classes of graphs closed under taking induced subgraphs are  the classes of disk intersection
graphs and unit disk intersection graphs. It is open whether \IFPC\ or any
other logic captures \PTIME\ on these classes. A very interesting and rich family
of classes of graphs closed under taking induced subgraphs is the family of
classes of graphs of bounded rank width \cite{oumsey06}, or equivalently, bounded clique width \cite{couola00}. It is conceivable that
\IFPC\ captures polynomial time on all classes of bounded rank width. To the best of my
knowledge, currently it is not even known whether isomorphism testing for
graphs of bounded rank width is in polynomial time.

\subsection*{Acknowledgements}
I would like to thank Yijia Chen and Bastian Laubner for valuable
comments on an earlier version of this paper.

\bibliographystyle{plain}
\bibliography{chordal}

\end{document}
